%% file: main.tex
\documentclass[11pt]{article}

\usepackage{fullpage}
\usepackage{amsthm,amssymb,amsmath}
\usepackage{hyperref}
\usepackage{cleveref}
\usepackage{aliascnt}

\newtheorem{theorem}{Theorem}

\newaliascnt{definition}{theorem}
\newaliascnt{observation}{theorem}
\newaliascnt{claim}{theorem}

\newtheorem{definition}[definition]{Definition}
\newtheorem{observation}[observation]{Observation}
\newtheorem{claim}[claim]{Claim}
\newenvironment{claimproof}{\noindent {\it Proof.}}{\hfill $\qed$}

\aliascntresetthe{definition}
\aliascntresetthe{observation}
\aliascntresetthe{claim}

\input{macros}

\begin{document}

\title{Computational Hardness of Private Coreset}

\author{Badih Ghazi\\Google Research\\{\small \texttt{badihghazi@gmail.com}}
\and
Crist\'obal Guzm\'an\\Pontificia Universidad Cat\'olica de Chile\\{\small \texttt{crguzman@google.com}}
\and
Pritish Kamath\\Google Research\\{\small \texttt{pritish@alum.mit.edu}}
\and
Alexander Knop\\Google Research\\{\small \texttt{alexanderknop@google.com}}
\and
Ravi Kumar\\Google Research\\{\small \texttt{ravi.k53@gmail.com}}
\and
Pasin Manurangsi\\Google Research\\{\small \texttt{pasin@google.com}}
}
\date{\today}

\maketitle

\begin{abstract}
    We study the problem of differentially private (DP) computation of coreset for the $k$-means objective. For a given input set of points, a coreset is another set of points such that the $k$-means objective for any candidate solution is preserved up to a multiplicative $(1 \pm \alpha)$ factor (and some additive factor). 
    
    We prove the first computational lower bounds for this problem. 
    Specifically, assuming the existence of one-way functions, we show that no polynomial-time 
    $(\eps, 1/n^{\omega(1)})$-DP algorithm can compute a coreset for $k$-means in the 
    $\ell_\infty$-metric for some constant $\alpha > 0$ (and some constant additive factor), even for $k=3$. 
    For $k$-means in the Euclidean metric, we show a similar result but only for $\alpha = \Theta\left(1/d^2\right)$, 
    where $d$ is the dimension.
\end{abstract}

\section{Introduction}

The widespread collection and use of personal data has necessitated rigorous frameworks for preserving individual privacy. 
Differential Privacy (DP), introduced by Dwork et al.~\cite{dwork2006calibrating,DworkKMMN06},
has established 
itself as the gold standard framework, enabling data release while protecting the confidentiality
of users. Roughly speaking, DP requires that the output distributions of the algorithm on two \emph{neighboring} input datasets---those that differ on a single user's input---to be similar. The  similarity is measured by two parameters $\eps, \delta \geq 0$; the standard setting, which will also be our focus, is $\eps = O(1)$ and $\delta = 1/n^{\omega(1)}$, where $n$ denotes the input dataset size (\Cref{def:dp}). 
The DP framework has been successfully  deployed in various large-scale applications, such as those by the U.S.
Census Bureau~\cite{abowd2018us} and Google's RAPPOR~\cite{erlingsson2014rappor}.


Among the myriad of data analysis tasks, clustering remains a fundamental primitive in unsupervised learning. 
Consequently, a significant body of research has been dedicated to designing clustering algorithms that satisfy DP~\cite{blum2005practical,GuptaLMRT10,BalcanDLMZ17,feldman2017coresets,NissimSV16,NissimS18,huang2018optimal,Stemmer20,GhaziKM20,ChangGKM21,GhaziHK0MNV23,TourHS24}. 
While most of these works focus on developing efficient algorithms that output a final set of cluster centers, many of them employ a private \emph{coreset} as a subroutine.

Recall that a coreset is another (possibly weighted) set of points such that the clustering (e.g., $k$-means or $k$-median) objective for any candidate solution is preserved up to a multiplicative $(1 \pm \alpha)$ factor and some additive factor $\pm \beta$ (\Cref{def:coreset}).\footnote{As formalized in \Cref{def:coreset}, the additive error $\beta$ in our notation is normalized so that it always lies in $[0, 1]$, unlike some previous work (e.g., \cite{FeldmanFKN09}) where $\beta$ is unnormalized and can be as large as $n$.}  In non-private settings, coreset computation has been extensively studied with the focus usually on finding a coreset of small size \cite{HarPeledM04,FrahlingS05,HarPeledK07,Chen09,LangbergS10,FeldmanL11,FichtenbergerGSSS13,FeldmanSS20,SohlerW18,BecchettiBC0S19,HuangV20,CohenAddadLSSS22,CohenAddadLSS22,CohenAddadD0SS25}; from this perspective, a coreset can be viewed as a succinct summary of the original dataset. 
Coresets are highly useful in data analysis since they allow ``big data'' to be replaced by a ``small summary'', thereby drastically reducing the computational cost of downstream algorithms. 
Furthermore, since the coreset approximates the cost function for any candidate solution, it enables hyperparameter tuning (e.g., $k$ in $k$-means) and trying multiple heuristics (e.g., different initializations in $k$-means) efficiently, 
all without revisiting the massive original dataset.

In the context of privacy, a coreset offers a distinct advantage besides efficiency: it decouples the privacy cost from the analysis.
Once a private coreset is computed, an analyst can run any clustering algorithm (e.g., $k$-means, $k$-median) on the coreset as many times as desired (for a fixed number of centers), without consuming additional privacy budget. 
Feldman et al.~\cite{FeldmanFKN09} first introduced the notion of private coresets and showed that they can be 
constructed information-theoretically with small errors. In particular, they give\footnote{Strictly speaking, the bound stated in~\cite[Corollary 3.3]{FeldmanFKN09} is $\beta = \tO_{\alpha}\Paren{\frac{k^2 d^2 \log n}{\eps n}}$ and is only for the Euclidean metric. The bound we claim here is by replacing the use of \cite{Chen09} with \cite{CohenAddadSS21}, which works for any metric space with doubling dimension $d$ (including the $\ell_\infty$-metric).} an $\eps$-DP algorithm for computing $k$-median and $k$-means coresets with $\beta = \tO_{\alpha}\Paren{\frac{k d}{\eps n}}$, where $d$ denotes the dimension. Alas, their algorithm is based on the exponential mechanism~\cite{McSherryT07}, which is not efficient; specifically, its running time is exponential in both $k$ and $d$. They also give a more efficient algorithm that removes the exponential dependence on $k$ in the running time, but both the running time and the additive error $\beta$ now become exponential in $d$. Subsequent works have explored different DP coreset algorithms for $k$-means in various settings \cite{GhaziKM20,ChangGKM21,TourHS24}; yet they all have this exponential dependence and efficient construction has remained a major open challenge. 




In this work, we investigate if this computational efficiency barrier is fundamental. Namely,
\begin{quote}\centering\slshape
    Does there exist a polynomial-time DP algorithm that returns\\[-0.5mm]
    a coreset with any approximation factor arbitrarily close to 1?
\end{quote}

We remark that, for private coresets, the requirement of small size is not critical: 
given a privately generated coreset of arbitrary size, one can apply a non-private coreset construction scheme
to compress it as a post-processing step. Due to this, we will not enforce the size restriction in the remainder of the paper. Note that this only strengthens our lower bounds.

\subsection{Our Contributions}

We provide a negative answer to the above question, establishing computational hardness results for private coreset construction.
By reducing from the hardness of generating synthetic data~\cite{UllmanV20}, we show that efficient private coreset
estimation cannot be achieved with standard cryptographic assumptions.

Our main results are the following:
\begin{itemize}
    \item \textbf{\boldmath Hardness in the $\ell_\infty$-metric\footnote{While $k$-means was first studied in the Euclidean metric, it has by now been well studied on other metrics, including $\ell_\infty$ metric. See e.g. \cite{CohenAddadS19} and references therein.}:} We prove that, assuming the existence of one-way functions\footnote{Since we do not deal with this assumption directly, we refrain from defining one-way functions and discussing their importance. Nevertheless, we note that this is a central assumption in Cryptography and detailed discussion can be found in any standard textbook on the topic (e.g. \cite{Goldreich2001}).}, 
    no polynomial-time $(\eps, \delta)$-DP algorithm can construct a coreset for $k$-means in the $\ell_\infty$-metric that 
    achieves an approximation factor $1 \pm \alpha$ for some constant $\alpha > 0$, even for $k = 3$ (\Cref{thm:main-linf}).
    
    \item \textbf{Hardness in the Euclidean  metric:} 
    We extend our analysis to the Euclidean setting and show that for the $\ell_2$-metric, any efficient private coreset 
    algorithm must incur an approximation error that scales inverse polynomially with the dimension. 
    Specifically, we rule out polynomial-time algorithms achieving an approximation factor better than $1 \pm \Theta(1/d^2)$ (\Cref{thm:main-l2}).
\end{itemize}

Our results show a dichotomy: while private coresets exist information-theoretically, 
they are computationally difficult to find under standard hardness assumptions. We further note that without the coreset size restriction, the non-private version of the problem is trivial, as one can simply output the input dataset. Thus, the computational challenges we present here are truly unique to the private setting.

\subsection{Technical Overview}
\label{subsec:overview}

To show computational hardness, we use a reduction from
the hardness of privately generating synthetic data for 3-literal
disjunctions~\cite{UllmanV20}.  Roughly speaking, \cite{UllmanV20} shows that, for 3-literal disjunction queries, it is hard to 
sanitize a dataset even in the case where we ``promise'' that some subset of these queries are 
always evaluated to one (``satisfiable'') in the input and that we only wish to be accurate on these queries. 

Our reduction is somewhat straightforward: We use the same dataset as the above problem (up to scaling)! 
The main property we need to show here is that, if we can find a DP coreset of such a dataset, 
then we can use the coreset to construct DP synthetic dataset for the 3-literal disjunction sanitization. 
Again, this latter step is somewhat straightforward: Given a coreset, we simply round each coordinate to $-1$ (false)
or $+1$ (true) based on their signs. This completes our reduction.

The main challenge however is to show that this is a valid reduction. 
Namely, that the output gives accurate estimates of the satisfiable 3-literal disjunctions. 
It turns out that this can be viewed as a 3-means cost query. 
Namely, suppose that we have a disjunction consisting of $x_{i_1}, x_{i_2}, x_{i_3}$ with signs
$s_{i_1}, s_{i_2}, s_{i_3}$ respectively. 
Then, we can create three centers, where the $j$th center only has the $i_j$ coordinate set to
$2s_{i_j}$ (and other coordinates being zero). 
Assuming that each output point $\tx$ belongs to $\{\pm 1\}^d$, it is simple to verify that the
distance is small only when the clause is satisfied. Notice that the ``gap'' between satisfiable
and unsatisfiable in the $\ell_\infty$-metric is constant (i.e., 3) whereas that of the $\ell_2$-metric depends on $d$ (i.e., $1 + \Theta(1/d)$). 
This is the reason that we get a constant $\alpha$ in \Cref{thm:main-linf} but requires $\alpha$ 
that decreases with $d$ in \Cref{thm:main-l2}. 
This completes the high-level overview of our reduction; in the actual proof, we of course cannot 
 assume that each $\tx$ belongs to $\{\pm 1\}^d$ and a significant effort is required to deal with this issue.

\subsection{Related Work}

Computational lower bounds intrinsic to DP are scarce. Dwork et al.~\cite{dwork2009complexity} established such lower bounds for generating synthetic data  when either the data universe or the number of queries is large (here large means exponential in some underlying size parameter). In particular, they provide a class of queries such that, assuming the existence of one-way functions, no polynomial-time $(\eps, 1/n^{\omega(1)})$-DP can generate a synthetic dataset that closely matches the query values of the input dataset. 
However, the class of queries produced in their arguments is specifically crafted based on a signature scheme, and does not correspond to a naturally studied query class. Subsequently, Ullman and Vadhan \cite{UllmanV20} extended the hardness to a large class of queries, including natural ones such as 2-way marginal queries. Their construction is based on the NP-hardness (via Karp reductions) of constraint satisfaction problems (CSPs), which focuses on discrete domains. From this perspective, our work expands the computational hardness landscape of private algorithms to the case of continuous domains, which we hope will stimulate further studies along this direction.

\section{Preliminaries}

Let $\bone_S \in \{0, 1\}^n$ denote the indicator vector of a set $S \subseteq [n]$; for ease, we
use $\bone_i$ 
for $\bone_{\{i\}}$.
For $v \in \R$, let $\sign(v) \in \{-1, 1\}$ denote the sign of $v$ (where $\sign(0) = 1$); moreover, for $x \in \R^d$, let $\sign(x) \in \{-1, 1\}^d$ denote the coordinate-wise sign of $x$, i.e. $\sign(x) = (\sign(x_1), \dots, \sign(x_d))$.
%

Let $\cX$ be a domain. 
For a query $f: \cX \to [0, 1]$ and a dataset $D = (x_1, \dots, x_n) \in \cX^n$, 
let the linear query $f(D)$ be defined as $f(D) := \frac{1}{n} \sum_{i \in [n]} f(x_i) = \E_{x \sim D}[f(x)]$.

\subsection{Differential Privacy}

We quickly recall the definition of differential privacy (DP) here. Two datasets $D, D' \in \cX^*$ are \emph{neighbors} iff they are of the same size and they differ on a single value, i.e., $D = (x_1, \dots, x_n), D' = (x'_1,\dots,x'_n) \in \cX^n$ and there exists $i^* \in [n]$ such that $x_i = x'_i$ for all $i \ne i^*$.

\begin{definition}[Differential Privacy \cite{dwork2006calibrating,DworkKMMN06}] \label{def:dp}
    Let $\eps \geq 0, \delta \in [0, 1]$.
    An algorithm $\alg: \cX^* \to \cO$ is \emph{$(\eps, \delta)$-differentially private (i.e., $(\eps, \delta)$-DP)} iff, for any two neighboring datasets $D, D'$ and for any subset $S \subseteq \cO$ of outputs, we have $$\Pr[\alg(D) \in S] \le e^\eps \cdot \Pr[\alg(D') \in S] + \delta.$$
\end{definition}

\subsection{Synthetic Data Generation}

As alluded to above, both the starting point of our reduction and the coreset problem itself can be viewed as synthetic data generation problems. An algorithm for generating a synthetic data is also referred to as a ``sanitizer''. More formally, a \emph{sanitizer} $\alg$ is a randomized algorithm that takes in the dataset $D \in \cX^*$ and
outputs another dataset $\tD \in \cX^*$.


\subsubsection{Accuracy and Efficiency}

The accuracy of the output dataset is often measure against some family $\cF$ of linear queries $f: \cX \to [0, 1]$; the following is a standard definition used in literature (e.g.~\cite{dwork2009complexity,UllmanV20}).

\begin{definition}[Accuracy] \label{def:acc-standard}
    A dataset $\tD$ is an \emph{$(\cF, \alpha)$-accurate} estimate of $D$ iff 
    $|f(\tD) - f(D)| \leq \alpha$ for all $f \in \cF$.     A sanitizer $\alg$ is \emph{$(\cF, \alpha, n)$-accurate} if for any dataset\footnote{We note that the failure probability is usually parameterized in the accuracy notation but here we fix it to $2/3$ for simplicity, since we are focusing on proving lower bounds. We note that the probability can be easily reduced by employing DP hyperparameter tuning (e.g. \cite{LT19,PapernotS22,GKKKMZ25}).} $D \in \cX^n$, \\ $\Pr_{\tD \sim \alg(D)}[\tD \text{ is an } (\cF, \alpha)\text{-accurate estimate of } D] \geq \frac{2}{3}$.  
\end{definition}

Throughout our work, we will consider the setting where $\cX = \cX_d \subseteq \R^d$ and, thus, the family of queries $\cF_d$ is also parameterized by $d$. We say that the sanitizer runs in polynomial time if it runs in $(nd)^{O(1)}$ time. Furthermore, the accuracy guarantee has to be against \emph{all values of $d \in \N$}. In this setting, we define a \emph{parameterized family} of queries to be a family $\cF = \bigcup_{d \in \N} \cF_d$ where $\cF_d$ is the family of queries corresponding to $\cX_d \subseteq \R^d$. We can define the accuracy for a parameterized family of queries as follows.

\begin{definition}[Accuracy w.r.t. Parameterized Family of Queries] \label{def:acc-parameterized}
For a parameterized family of queries $\cF = \bigcup_{d \in \N} \cF_d$ where $\cF_d$ is the family of queries corresponding to $\cX_d \subseteq \R^d$, we say that a sanitizer $\alg$ is $(\cF, \alpha)$-accurate if the following holds: There exists a constant $\nu > 0$ such that, for all $d \in \N$ and $n \geq \Theta(d^\nu)$, $\alg$ is $(\cF_d, \alpha, n)$-accurate.
\end{definition}



%

\subsubsection{Promise Sanitizer}

To prove our results, we will use the hardness result from~\cite[Theorem 4.4]{UllmanV20},  
which actually holds even against a weaker notion which we name \emph{promise sanitizers}. In words, promise sanitizers only provide the accuracy guarantee on the queries $f \in \cF$ whose values are one on the entire input dataset. For the queries that violate this condition, promise sanitizers do not provide any accuracy guarantee. This is formalized below.
 
\begin{definition}[Promise sanitizer] \label{def:prom-sanitizer}
    Let $\cF$ be a class of queries and $\gamma \in [0, 1]$.  $\alg$ is an  \emph{$(\cF, \gamma, n)$-promise sanitizer} if
    $\Pr_{\tD \sim \alg(D)}[f(\tD) \geq \gamma \text{ for all } f \in \tcF] \geq 2 / 3$, 
    for any dataset $D \in \cX^n$ and $\tcF \subseteq \cF$ such that $f(D) = 1$ for all $f \in \tcF$.

    Similar to \Cref{def:acc-parameterized}, for a parameterized family of queries $\cF = \bigcup_{d \in \N} \cF_d$, we say that $\alg$ is an $(\cF, \gamma)$-promise sanitizer if there exists a constant $\nu > 0$ such that, for all $d \in \N$ and $n \geq \Theta(d^\nu)$, $\alg$ is an $(\cF_d, \gamma, n)$-promise sanitizer.
\end{definition}
Note that a $(\cF, 1 - \gamma)$-accurate sanitizer is a $(\cF, \gamma)$-promise sanitizer. In other words, a promise sanitizer is a weaker notion compared to an accurate sanitizer. While Ullman and Vadhan~\cite{UllmanV20} only stated their results in term of an accurate sanitizer, we observe that their proof of~\cite[Theorem 4.4]{UllmanV20} already yields the hardness for a promise sanitizer as well\footnote{See Appendix~\ref{app:promise-sanitizer} for more detail.}.


To state their result, let $\cFthreesat_d$ denote the class of all 3-literal disjunctions where the domain is $\{\pm 1\}^d$ 
(where we think of $-1$ as ``false'' and $+1$ as ``true''). More formally, for every (literal indices) $\bi = (i_1, i_2, i_3) \in [d]^3$ and (signs) $\bs = (s_1, s_2, s_3) \in \{\pm 1\}^3$, we define $\psi_{\bi, \bs}: \{\pm 1\}^d \to \{\pm 1\}$ as follows:
\begin{align*}
\psi_{\bi, \bs}(x) = 1 - \frac{1}{4}\cdot (1 - s_1 x_{i_1})(1 - s_2 x_{i_2})(1 - s_3 x_{i_3})
\end{align*}
In other words, this is the evaluation of the 3-literal disjunction $s_1 x_{i_1} \vee s_2 x_{i_2} \vee s_3 x_{i_3}$ on $x$. We then let $\cFthreesat_d := \{\psi_{\bi, \bs} \mid \bi \in [d]^3, \bs \in \{\pm1\}^3\}$, and let $\cFthreesat$ be the parameterized family $\bigcup_{d \in \N} \cFthreesat_d$. The hardness result we need is as follows.

\begin{theorem}[\cite{UllmanV20}] \label{thm:3sat-hardness}
    Let $\eps > 0$ be any constant. Assuming the existence of one-way functions, there is no polynomial-time 
    $(\eps, 1/n^{\omega(1)})$-DP $(\cFthreesat, \gamma)$-promise sanitizer for some constant $\gamma \in (0, 1)$.
\end{theorem}

\subsection{Coreset Estimation}
Let $\bB_{p,d}$ denote a unit ball in the $\ell_p$-metric on $\R^d$; i.e., $\bB_{p,d} := \{x \in \R^d ~:~ \|x\|_p \leq 1\}$.
In the $\ell_p$-coreset problem, we consider the domain $\cX = \bB_{p,d}$. 
Let $\cFcoreset{p,d}{k}$ denote the family of coreset cost functions; i.e., 
$\cFcoreset{p,d}{k} = \{\cost_{\cC, p}: \bB_{p,d} \to [0, 1] ~\mid~ \cC \in \bB_{p,d}^k\}$,
where $\cost_{\cC, p}(x) = \min_{c \in \cC} \|x - c\|_p^2$. Finally, we let $\cFcoreset{p}{k}$ be the parameterized family $\bigcup_{d \in \N} \cFcoreset{p,d}{k}$.

%
\begin{definition}[$\ell_p$-coreset] \label{def:coreset}
    A dataset $\tD \in \bB_{p,d}^*$
     is a \emph{$(k, p, \alpha, \beta)$-coreset}
    of $D \in \bB_{p,d}^*$ if\footnote{We use the same notation as linear queries: For any dataset $D$, let $\cost_{\cC, p}(D) = \E_{x \sim D} [\cost_{\cC, p}(x)]$}
    \[
        (1 - \alpha) \cdot \cost_{\cC, p}(D) - \beta \leq \cost_{\cC, p}(\tD) \leq (1 + \alpha) \cdot \cost_{\cC, p}(D) + \beta,
    \]
    for all $\cC \in \bB_{p,d}^k$.
    An algorithm $\alg$ is a \emph{$(k, p, \alpha, \beta, d, n)$-coreset estimator} iff for any dataset $D \in \bB_{p,d}^n$, it holds that 
    \[
        \Pr_{\tD \sim \alg(D)}[\tD \text{ is a } (k, p, \alpha, \beta)\text{-coreset of } D] \ge 2 / 3.
    \]

    Similar to \Cref{def:acc-parameterized}, we say that $\alg$ is a \emph{$(k, p, \alpha, \beta)$-coreset estimator} if there exists a constant $\nu > 0$ such that, for all $d \in \N$ and $n \geq \Theta(d^\nu)$, $\alg$ is a $(k, p, \alpha, \beta, d, n)$-coreset estimator.
\end{definition}


As explained in the introduction, we do not require that the private coreset be small. We also note that, unlike standard notations of coreset, we also require the coreset to be unweighted for notational convenience. Since there is no restriction on the coreset size and we allow an additive error $\beta$, this is w.l.o.g. because we can always rescale the weights by a large factor (e.g. $O(|\tD|/\beta)$) and round them to integers to obtain an unweighted dataset.

Also, to avoid dealing with two parameters $\alpha, \beta$, it will be more convenient to work with the accuracy definition of the sanitizer (\Cref{def:acc-standard,def:acc-parameterized}) instead. Indeed, it is simple to see that any algorithm that solves $(k, p, \alpha, \beta)$-coreset is also an accurate sanitizer for the query family $\cFcoreset{p}{k}$, as formalized below.

\begin{observation} \label{obs:coreset-to-sanitizer}
    Any algorithm that is a $(k, p, \alpha, \beta)$-coreset estimator is an 
    $\left(\cFcoreset{p}{k}, 4\alpha + \beta\right)$-accurate sanitizer.
\end{observation}

\begin{proof}
    This follows immediately from the definition of coreset, since $\cost_{\cC, p}(D) \in [0, 2]$.
\end{proof}

\section{\boldmath Hardness in the \texorpdfstring{$\ell_{\infty}$}{L-infty}-metric}

In this section, we prove the hardness of computing DP coreset in the $\ell_\infty$-metric:
\begin{theorem} \label{thm:main-linf}
    Let $\eps > 0$ be any constant. Assuming the existence of one-way functions, there is no polynomial-time 
    $(\eps, 1/n^{\omega(1)})$-DP algorithm for the $(3, \infty, \alpha, \beta)$-coreset estimation
    problem for some constant $\alpha, \beta > 0$.
\end{theorem}

As alluded to earlier, it is more convenient to prove the hardness in terms of sanitizer, as stated in \Cref{thm:linf-sanitizer} below. 
Due to \Cref{obs:coreset-to-sanitizer}, \Cref{thm:main-linf} follows as an immediate corollary of \Cref{thm:linf-sanitizer}.

\begin{theorem} \label{thm:linf-sanitizer}
    Let $\eps > 0$ be any constant. Assuming the existence of one-way functions, there is no polynomial-time $(\eps, 1/n^{\omega(1)})$-DP 
    $(\cFcoreset{\infty}{k}, \xi)$-accurate sanitizer for $k = 3$ and some constant $\xi \in[0,1)$.
\end{theorem}

The proof of this theorem follows closely the outline in \Cref{subsec:overview}. Namely, for an input dataset $D$ for the $(\cFthreesat, \gamma)$-promise sanitization problem (where each $x \in D$ belongs to $\{-1,1\}^d$), we simply construct a dataset $D'$ which is $D$ scaled by a factor of $\kappa = \frac{1}{2}$. We then pass $D'$ to the coreset estimator and round it to obtain the final output. 

As explained in \Cref{subsec:overview}, for each $\psi \in \cFthreesat$, we can construct a 3-means query--where the centers are denoted by $\cC^{\psi}$ in the proof below--to check if $\psi(D) = 1$. Since we do not have any guarantee that the coreset only contains points in $\{-\kappa, \kappa\}^d$, we additionally consider the origin as a center to prevent any ``cheating''. Intuitively, if some coordinates have absolute values larger than $\kappa$, then its distance to the origin would be too large. Otherwise, if the absolute values are smaller than $\kappa$, then its distance to $\cC^{\psi}$ is too large. This is argued formally below in \Cref{claim:linf-lb-unsat}.

\begin{proof}[Proof of \Cref{thm:linf-sanitizer}]
    We prove this by a reduction from the $(\cFthreesat, \gamma)$-promise sanitization problem, 
    which is known to be hard from \Cref{thm:3sat-hardness}.  
    Assume for the sake of contradiction that there exists a polynomial-time $(\eps, 1/n^{\omega(1)})$-DP $(\cFcoreset{\infty}{k}, \xi)$-sanitizer $\alg$, where $\xi = \frac{1 - \gamma}{4}$.
    
    We construct the $(\cFthreesat, \gamma)$-promise sanitizer $\newalg$ as follows.  Let $\kappa = \frac{1}{2}$.
    \begin{itemize}
        \item Let $D \in (\{ \pm 1\}^d)^n$ be the input.
        \item Construct dataset $D' = \{ x' = \kappa \cdot x \mid x \in D \}$. 
        \item Run $\alg(D')$ to get an output $\tD'$.
        \item Construct dataset $\tD = \{ \tx = \sign(\tx') \mid \tx' \in \tD' \}$.
        \item Output $\tD$.
    \end{itemize}
    Since $\newalg$ is a post-processing of $\alg$, the $(\eps, 1/n^{\omega(1)})$-DP guarantee continues to hold. 
    
    Next, we will show that $\newalg$ solves the $(\cFthreesat, \gamma)$-promise sanitization problem. To do this, assume (the promise) that $\psi(D) = 1$ for all $\psi \in \tcF \subseteq \cFthreesat$. Since $\alg$ is an $(\cFcoreset{\infty}{k}, \xi)$-accurate sanitizer, with probability at least $2/3$, the following holds for all $\cC \in \bB_p^k$:
    \begin{align} \label{eq:linf-upperbound}
        \left|\cost_{\cC, \infty}(\alg(D')) - \cost_{\cC, \infty}(D')\right|
        =
        \left|\cost_{\cC, \infty}(\tD') - \cost_{\cC, \infty}(D')\right| 
        \leq \xi.
    \end{align}
    We will show that, when the above holds, we have $\psi(\tD) \geq \gamma$ for all $\psi \in \tcF$.
    
    Let $\cC^{\bzero} = \{\bzero\}$. 
    Furthermore, consider any 3-literal disjunction $\psi = \psi_{\bi, \bs}$ in $\tcF$. 
    Recall that the indices of the three literals are $i_1, i_2, i_3$ and their signs are
    $s_1, s_2, s_3$ respectively. 
    Consider centers $c^{\psi, 1}, c^{\psi, 2}, c^{\psi, 3}$ defined by 
    $c^{\psi, j} = 2\kappa \cdot s_j \cdot \bone_{i_j}$. 
    Finally, let $\cC^{\psi} = (c^{\psi, 1}, c^{\psi, 2}, c^{\psi, 3})$.
    
    There are two claims crucial to the proof. 
    The first is a claim that upper bounds $\cost_{\cC^{\bzero}, \infty}(x') + \cost_{\cC^{\psi}, \infty}(x')$ 
    for $x' \in D'$:
    
    \begin{claim} \label{claim:linf-ub-sat}
        For any $x' \in D'$ and $\psi \in \tcF$, we have
        $\cost_{\cC^{\bzero}, \infty}(x') + \cost_{\cC^{\psi}, \infty}(x') \leq 2\kappa^2.$
    \end{claim}
    
    \begin{claimproof}
        Notice that $\cost_{\cC^{\bzero}, \infty}(x') = \|x'\|^2_{\infty} = \kappa^2$ for all $x' \in D'$. 
        
        Meanwhile, from the 
        assumption that $\psi(D) = 1$, we also have that $\psi(x) = 1$ for all $x \in D$. 
        That is, there exists a literal of $\psi$, say the $j$th literal for $j \in [3]$, that is set to true in $x$ (i.e. $x_{i_j} = s_j$). 
        Consider $x' - c^{\psi, j}$. We claim that all of its coordinates have absolute value $\kappa$, i.e. $x' - c^{\psi, j} \in \{\pm \kappa\}^d$. To see this, consider two cases based on the coordinate $i \in [d]$:
        \begin{itemize}
        \item Case I: $i \ne i_j$. In this case, we simply have $(x' - c^{\psi, j})_i = x'_i \in \{\pm \kappa\}$ as desired.
        \item Case II: $i = i_j$. In this case, $(x' - c^{\psi, j})_{i_j} = \kappa \cdot x_{i_j} -  2\kappa \cdot s_j = - \kappa \cdot x_{i_j} \in \{\pm \kappa\}$.
        \end{itemize}
        Hence, we have $x' - c^{\psi, j} \in \{\pm \kappa\}^d$, which implies that $\|x' - c^{\psi, j}\|_{\infty} = \kappa$. 
        Thus, $\cost_{\cC^{\bzero}, \infty}(x') + \cost_{\cC^{\psi}, \infty}(x') \leq \kappa^2 + \kappa^2 = 2\kappa^2$.
    \end{claimproof}
    
    The second claim is a lower bound on $\cost_{\cC^{\bzero}, \infty}(\tx') + \cost_{\cC^{\psi}, \infty}(\tx')$ for $\tx' \in \tD'$, as stated below. It should be noted that, when $\psi(\tx) = 1$, this lower bound is exactly the same as the bound in the above claim. However, if $\psi(\tx) = 0$, the lower bound here is larger. This will indeed allow us to bound the number of $\tx$'s that fall into the latter case.
    
    \begin{claim} \label{claim:linf-lb-unsat}
        For any $\tx' \in \tD'$ and $\psi \in \tcF$, we have 
        $\cost_{\cC^{\bzero}, \infty}(\tx') + \cost_{\cC^{\psi}, \infty}(\tx') \geq 2\kappa^2(2 - \psi(\tx)).$
    \end{claim}
    
    \begin{claimproof}
        First, consider the case $\psi(\tx) = 1$. 
        We have $\cost_{\cC^{\psi}, \infty}(\tx') = \|c^{\psi, j} - \tx'\|^2_{\infty}$ for some $j \in [3]$. 
        We can lower bound this further by observing that 
        $\|c^{\psi, j} - \tx'\|_{\infty} \geq |(c^{\psi, j})_{i_j} - \tx'_{i_{j}}| = |2\kappa \cdot s_{j} - \tx'_{i_{j}}|$. 
        Meanwhile, $\cost_{\cC^{\bzero}, \infty}(\tx') = \|\tx'\|_\infty^2 \geq (\tx'_{i_{j}})^2$. 
        Thus, we have
        \begin{align*}
            \cost_{\cC^{\bzero}, \infty}(\tx') + \cost_{\cC^{\psi}, \infty}(\tx') \geq \left(2\kappa \cdot s_{j} - \tx'_{i_{j}}\right)^2 + (\tx'_{i_{j}})^2 &= 2\kappa^2 + 2\left(\kappa \cdot s_{j} - \tx'_{i_{j}}\right)^2 \\ &\geq 2\kappa^2.
        \end{align*}
        Next, suppose that $\psi(\tx) = 0$. In this case, for all $j \in [3]$, we have that $(c^{\psi, j})_{i_{j}}$ and $\tx'_{i_{j}}$ have different signs. Thus, $\|c^{\psi, j} - \tx'\|_{\infty} \geq |(c^{\psi, j})_{i_{j}} - \tx'_{i_{j}}| \geq |(c^{\psi, j})_{i_{j}}| = 2\kappa$. As a result, we have $\cost_{\cC^{\psi}, \infty}(\tx') \geq (2\kappa)^2 = 4\kappa^2$ as desired.
    \end{claimproof}
    
    Combining the above two claims and \eqref{eq:linf-upperbound}, we get
    \begin{align*}
    2\kappa^2(2 - \psi(\tD)) &= \E_{\tx \sim \tD} \left[2\kappa^2(2 - \psi(\tx))\right] \\
    \text{(\Cref{claim:linf-lb-unsat})} &\leq  \E_{\tx' \sim \tD'} \left[\cost_{\cC^{\bzero}, \infty}(\tx') + \cost_{\cC^{\psi}, \infty}(\tx')\right] \\
    &= \cost_{\cC^{\bzero}, \infty}(\tD') + \cost_{\cC^{\psi}, \infty}(\tD') \\
    \text{Using \eqref{eq:linf-upperbound}} &\leq \cost_{\cC^{\bzero}, \infty}(D') + \cost_{\cC^{\psi}, \infty}(D') + 2\xi \\
    \text{(\Cref{claim:linf-ub-sat})} &\leq 2\kappa^2 + 2\xi.
    \end{align*}
    Thus, from our choice of parameter $\xi = \frac{1 - \gamma}{4}, \kappa = \frac{1}{2}$, we have that $\psi(\tD) \geq \gamma$. As a result, we can conclude that $\newalg$ solves the $(\cFthreesat, \gamma)$-promise sanitization problem. Finally, by \Cref{thm:3sat-hardness}, this cannot hold assuming the existence of one-way functions.
\end{proof}

\section{Hardness in the Euclidean metric}

Next, we prove that hardness in the Euclidean metric, which is similar to the hardness for the $\ell_\infty$-metric (\Cref{thm:main-linf}) except with $\alpha, \beta = \Omega\left(\frac{1}{d^2}\right)$ instead of constants.
\begin{theorem} \label{thm:main-l2}
    Let $\eps > 0$ be any constant. Assuming the existence of one-way functions, there is no polynomial-time
    $(\eps, 1/n^{\omega(1)})$-DP algorithm for the $(3, 2, \alpha, \beta)$-coreset estimation problem for some $\alpha = \Omega\left(\frac{1}{d^2}\right)$ and 
    $\beta = \Omega\left(\frac{1}{d^2}\right)$.
\end{theorem}

Again, we actually prove the sanitizer variant of the theorem, as stated below, from 
which \Cref{thm:main-l2} follows as a corollary (due to \Cref{obs:coreset-to-sanitizer}).

\begin{theorem} \label{thm:l2-sanitizer}
    Let $\eps > 0$ be any constant. Assuming the existence of one-way functions, there is no polynomial-time 
    $(\eps, 1/n^{\omega(1)})$-DP $(\cFcoreset{2}{k}, \xi)$-sanitizer for $k = 3$ and some $\xi = \Omega(1/d^2)$.
\end{theorem}

The proof follows a similar strategy as in \Cref{thm:linf-sanitizer} but it requires a more subtle sequence of inequalities to prove the accuracy bound. We note that we also need to set a smaller scaling factor $\kappa$: Setting $\kappa = \frac{1}{2}$ as before would result in $\|x'\|_2 = {\kappa\sqrt{d}} > 1$. This suggests setting $\kappa = \frac{1}{\sqrt{d}}$. However, we actually go with the setting $\kappa = \frac{1}{d}$ instead because, to facilitate one of our sum-of-squares rearrangement (in \Cref{claim:l2-deviation-from-boolean}), we need to select the centers (in $\cC^i$) to be of norm $d\kappa$, which necessitates $\kappa \leq \frac{1}{d}$.

\begin{proof}[Proof of \Cref{thm:l2-sanitizer}]
Similar to the proof of \Cref{thm:linf-sanitizer}, we reduce from $(\cFthreesat, \gamma)$-promise sanitization problem.  Assume for the sake of contradiction that there exists a polynomial-time $(\eps, 1/n^{\omega(1)})$-DP $(\cFcoreset{2}{k}, \xi)$-sanitizer $\alg$, where $\xi = \frac{(1 - \gamma)^2}{4d^2}$.

Let $\newalg$ be exactly the same as in the proof of \Cref{thm:linf-sanitizer} except that we set $x' = \kappa \cdot x$ where $\kappa = \frac{1}{d}$ (instead of $\kappa = 1/2$ as before), which ensures that $\|x'\|_2 \leq 1$. Again, $\newalg$ is $(\eps, 1/n^{\omega(1)})$-DP due to post-processing. We are left to show that $\newalg$ solves the $(\cFthreesat, \gamma)$-promise sanitization problem. To do this, assume that $\psi(D) = 1$ for all $\psi \in \tcF$. Since $\alg$ is an $(\cFcoreset{2}{k}, \xi)$-sanitizer, with probability at least $2/3$, the following holds for all $\cC \in \bB_2^k$:
\begin{align} \label{eq:l2-upperbound}
\left|\cost_{\cC, 2}(\tD') - \cost_{\cC, 2}(D')\right| \leq \xi.
\end{align}
We will show that, when the above holds, we have $\psi(\tD) \geq \gamma$ for all $\psi \in \tcF$.

Consider any 3-literal disjunction $\psi$. We use the same notation (e.g.,  $\cC^{\bzero}, \cC^{\psi}, c^{\psi, j}$) as in the proof of \Cref{thm:linf-sanitizer}. Additionally, for each $i \in [d]$, let $\cC^i = \{+d\kappa \cdot \bone_i, -d\kappa \cdot \bone_i\}$.

We begin by showing that most of the coordinates of $\tx' \in \tD'$ have to be close to $-\kappa$ or $+\kappa$. The exact bound is given below.

\begin{claim} \label{claim:l2-deviation-from-boolean}
$\E_{\tx' \in \tD'}[\|\tx' - \kappa \cdot \sign(\tx')\|_2^2] \leq \xi$.
\end{claim}

\begin{claimproof}
For any $y \in \R^d$, observe that $\cost_{\cC^{i}, 2}(y) = \|y -  d\kappa \cdot \sign(y_i) \cdot \bone_i\|^2 = (y_i - d\kappa \cdot \sign(y_i))^2 + \sum_{j \ne i} y_j^2$. Thus, aggregating this across all $i \in [d]$, we thus have
\begin{align*}
\sum_{i \in [d]} \cost_{\cC^{i},2}(y) &= \sum_{i \in [d]} \left((d - 1) y_i^2 + (y_i - d\kappa \cdot \sign(y_i))^2\right) \\ &= \sum_{i \in [d]} \left((d^2 - d)\kappa^2 + d(y_i - \kappa \cdot \sign(y_i))^2\right) \\
&= (d^3 - d^2)\kappa^2 + d \cdot \|y - \kappa \cdot \sign(y)\|_2^2.
\end{align*}
From \eqref{eq:l2-upperbound} and since $x' = \kappa \cdot \sign(x)$ for all $x \in D$, we thus have
\begin{align*}
d \cdot \xi \geq \sum_{i \in [d]} \left( \cost_{\cC^{i},2}(\tD') - \cost_{\cC^{i}, 2}(D') \right) = \E_{\tx' \in \tD'}[d \cdot \|\tx' - \kappa \cdot \sign(\tx')\|_2^2].
\end{align*}
Dividing both sides by $d$ yields the claimed inequality.
\end{claimproof}

The next two claims are in the same spirit as \Cref{claim:linf-ub-sat} and \Cref{claim:linf-lb-unsat} in the proof of \Cref{thm:linf-sanitizer}.

\begin{claim} \label{claim:l2-ub-sat}
For any $x' \in D'$ and $\psi \in \tcF$, we have
$\cost_{\cC^{\psi}, 2}(x') - \cost_{\cC^{\bzero}, 2}(x') \leq 0.$
\end{claim}

\begin{claimproof}
From our assumption, there exists a literal of $\psi$, say the $j$th literal for $j \in [3]$, that is set to true in $x$. From the proof of \Cref{claim:linf-ub-sat}, we have $x' - c^{\psi, j} \in \{\pm \kappa\}^d$. Recall that $x'$ also belongs to $\{\pm \kappa\}^d$. Thus, $\cost_{\cC^{\bzero}, 2}(x') = \|x'\|_2^2 = \kappa^2 \cdot d = \|x' - c^{\psi, j}\|_2^2 \geq \cost_{\cC^{\psi}, 2}(x')$.
\end{claimproof}

\begin{claim} \label{claim:l2-lb-unsat}
For any $\tx' \in \tD'$ and $\psi \in \tcF$, we have 
\begin{align*}
\cost_{\cC^{\psi}, 2}(\tx') - \cost_{\cC^{\bzero}, 2}(\tx') &\geq 4\kappa^2 (1 - \psi(\tx)) - 4\kappa \cdot \|\tx' - \kappa \cdot \sign(\tx')\|_2.
\end{align*}
\end{claim}
\begin{claimproof}
We have $\cost_{\cC^{\psi}, 2}(\tx') = \|c^{\psi, j} - \tx'\|_{2}^2$ for some $j \in [3]$. Thus,
\begin{align*}
\cost_{\cC^{\psi}, 2}(\tx') - \cost_{\cC^{\bzero}, 2}(\tx') &= \|c^{\psi, j} - \tx'\|_{2}^2 - \|\tx'\|_2^2 \\
&= \|c^{\psi, j}\|_2^2 -2\left<c^{\psi, j}, \tx'\right> \\
&= 4\kappa^2 - 4\kappa \cdot s_{j} \cdot \tx'_{i_{j}}.
\end{align*}
First, consider the case $\psi(\tx) = 1$. In this case, we can bound the above further by
\begin{align*}
\cost_{\cC^{\psi}, 2}(\tx') - \cost_{\cC^{\bzero}, 2}(\tx') &\geq 4\kappa^2 - 4\kappa \cdot |\tx'_{i_{j}}| \\
&= -4\kappa\left(|\tx'_{i_{j}}| - \kappa\right) \\
&\geq -4\kappa \cdot |\tx'_{i_{j}} - \kappa \cdot \sign(\tx'_{i_{j}})| \\
&\geq -4\kappa \cdot \|\tx' - \kappa \cdot \sign(\tx')\|_2,
\end{align*}
which implies the desired bound.

Finally, for the case $\psi(\tx) = 0$, we simply have $s_{j} \cdot \tx'_{i_{j}} \leq 0$ and thus $\cost_{\cC^{\psi}, 2}(\tx') - \cost_{\cC^{\bzero}, 2}(\tx')  \geq 4\kappa^2$.
\end{claimproof}

Now, from the above three claims, we have
\begin{align*}
0 &\geq \E_{x' \sim D'}\left[\cost_{\cC^{\psi}, 2}(x') - \cost_{\cC^{\bzero}, 2}(x')\right] \\
&= \cost_{\cC^{\psi}, 2}(D') - \cost_{\cC^{\bzero}, 2}(D') \\
\text{Using \eqref{eq:l2-upperbound}} &\geq -2\xi +  \cost_{\cC^{\psi}, 2}(\tD') - \cost_{\cC^{\bzero}, 2}(\tD') \\
&= -2\xi + \E_{\tx' \sim \tD'}\left[\cost_{\cC^{\psi}, 2}(\tx') - \cost_{\cC^{\bzero}, 2}(\tx')\right] \\
(\text{\Cref{claim:l2-lb-unsat}}) &\geq -2\xi + \E_{\tx' \sim \tD'}\left[4\kappa^2 (1 - \psi(\tx)) - 4\kappa \cdot \|\tx' - \kappa \cdot \sign(\tx')\|_2\right] \\
&= 4\kappa^2(1 - \psi(\tD)) - 2\xi - 4\kappa \cdot \E_{\tx' \sim \tD'}\left[\|\tx' - \kappa \cdot \sign(\tx')\|_2\right] \\
(\text{Cauchy–Schwarz Inequality}) &\geq 4\kappa^2(1 - \psi(\tD)) - 2\xi - 4\kappa \cdot \sqrt{\E_{\tx' \sim \tD'}\left[\|\tx' - \kappa \cdot \sign(\tx')\|_2^2\right]} \\
(\text{\Cref{claim:l2-deviation-from-boolean}}) &\geq  4\kappa^2(1 - \psi(\tD)) - 2\xi - 4\kappa \cdot \sqrt{\xi} \\
&\geq 4\kappa^2\left(1 - \frac{\xi}{2\kappa^2} - \frac{\sqrt{\xi}}{\kappa} - \psi(\tD)\right),
\end{align*}
where the first inequality is from \Cref{claim:l2-ub-sat}.

The above inequality implies that $\psi(\tD) \geq 1 - \frac{\xi}{2\kappa^2} - \frac{\sqrt{\xi}}{\kappa}$, which is at least $\gamma$ due to our choice of $\xi = \frac{(1 - \gamma)^2}{4d^2}, \kappa = \frac{1}{d}$.  As a result, $\newalg$ solves the $(\cFthreesat, \gamma)$-promise sanitization problem; by \Cref{thm:3sat-hardness}, this cannot hold assuming the existence of one-way functions.
\end{proof}

\section{Conclusion and Open Questions}
In this paper, we showed computational lower bounds for efficiently constructing DP coresets, assuming the one-way functions exist.   Our result should be contrasted against the work of Feldman et al.~\cite{FeldmanFKN09}, who showed that small private coresets exist information-theoretically.

Our work leaves open several questions for future research.  
While our hardness result for the $\ell_2$-metric rules out approximation factors $1 \pm \Theta(1/d^2)$, it is an interesting question if a dimension-independent
constant-factor hardness (similar to our $\ell_\infty$-metric result) can be obtained for the Euclidean metric. Additionally, our hardness results hold for $k = 3$ due to our ``gadget'' to reduce from 3-literal disjunctions, leaving the case $k = 2$ as an intriguing open question. We remark that, while Ullman and Vadhan~\cite{UllmanV20} also proved the hardness of generating synthetic data for 2-way marginals, their result does not achieve ``perfect completeness'' (in the promise sense that $f(D) = 1$ for all $f \in \cF$ in \Cref{def:prom-sanitizer}). This property is crucial in our reduction, preventing us from obtaining hardness for the $k = 2$ case. 

Furthermore, our proofs rely on the fact that we are dealing with the $k$-means objective (by rearranging terms into sum-of-squares); it would be interesting to see if our results can be extended to the case of, e.g., $k$-median as well. Finally, understanding the computational landscape of private coresets in other metric spaces or for other objective functions is an interesting research direction. 

\bibliographystyle{alpha}
\bibliography{main}

\appendix

\input{app-promise-sanitizer-explained}

\end{document}

%% file: macros.tex
\newcommand{\eps}{\varepsilon}
\newcommand{\cX}{\mathcal{X}}
\newcommand{\cO}{\mathcal{O}}
\newcommand{\tD}{\tilde{D}}
\newcommand{\tO}{\tilde{O}}
\newcommand{\tx}{\tilde{x}}
\newcommand{\alg}{\mathcal{A}}
\newcommand{\newalg}{\tilde{\mathcal{A}}}
\newcommand{\cost}{\mathrm{cost}}

\newcommand{\sign}{\mathrm{sgn}}
\newcommand{\cF}{\mathcal{F}}
\newcommand{\cC}{\mathcal{C}}
\newcommand{\cFcoreset}[2]{\mathcal{F}^{{#2}\mathrm{\text{-}means}}_{#1}}
\newcommand{\tcF}{\tilde{\cF}}
\newcommand{\bB}{\mathbb{B}}
\newcommand{\cFthreesat}{\cF^{\mathsf{3Disj}}}
\newcommand{\R}{\mathbb{R}}
\newcommand{\N}{\mathbb{N}}
\newcommand{\bzero}{\mathbf{0}}
\newcommand{\bone}{\mathbf{1}}
\newcommand{\E}{\mathbb{E}}
\newcommand{\Paren}[1]{\left(#1\right)}

\newcommand{\bi}{\mathbf{i}}
\newcommand{\bs}{\mathbf{s}}

%% file: app-promise-sanitizer-explained.tex
\newcommand{\Gen}{\mathtt{Gen}}
\newcommand{\Sign}{\mathtt{Sign}}
\newcommand{\Ver}{\mathtt{Ver}}
\newcommand{\sk}{\mathsf{sk}}
\newcommand{\vk}{\mathsf{vk}}
\newcommand{\arr}{\mathtt{R}}
\newcommand{\enc}{\mathtt{Enc}}
\newcommand{\dec}{\mathtt{Dec}}
\newcommand{\tn}{\tilde{n}}

\section{Hardness from \cite{UllmanV20} in terms of Promise Sanitizer} \label{app:promise-sanitizer}

In this section, we briefly explain how to interpret the hardness 
result of Ullman and Vadhan \cite{UllmanV20} in terms of a promise sanitizer (\Cref{thm:3sat-hardness}).  

Informally, the proof is by contradiction, assuming that an efficient DP promise sanitizer exists.  The 
idea is to construct a dataset by taking a set of valid digital signatures and encoding them using a PCP of the signature's verification circuit.  Now, if the sanitizer can produce a synthetic dataset that satisfies a good fraction of the PCP predicates, the PCP's soundness guarantee means we can decode a valid signature from the synthetic dataset.  But, this yields a contradiction: 
either the decoded signature is new, which would violate 
the security of the super-secure signature scheme, or it 
matches one of the original signatures, which would violate 
the DP guarantees.

To avoid repeating their proof verbatim, we will only sketch the high-level arguments here.

\subsection{Background}

We review a couple of key tools required for the construction of \cite{UllmanV20}.

\paragraph{Super-Secure Digital Signature Scheme.}
A \emph{Digital Signature Scheme (DSS)} consists of three polynomial-time algorithms:
\begin{itemize}
\item $\Gen$ takes in the security parameter $1^{\kappa}$ and output a secret key $\sk$ and a verification key $\vk$,
\item $\Sign$ takes in\footnote{We use a ``no-message'' version of the signature scheme, which can be initiated from the standard notion by setting the message $m$ to be fixed (e.g., empty string) for every signature.} a secret key $\sk$ and outputs a (random) signature $\sigma$,
\item $\Ver$ takes in the verification key $\vk$ and a signature $\sigma$. For any valid signature $\sigma$ produced by $\Sign$, $\Ver(\vk, \sigma)$ always accepts (i.e., returns 1).
\end{itemize}

The security guarantee is that, an adversary that is given the verification key $\vk$ together with a set $\Sigma$ of any polynomially large number of signatures, cannot produce a new signature $\sigma^* \notin \Sigma$ that is accepted by $\Ver$ with at least $1/n^{\Omega(1)}$ probability.

It is known that super-secure digital signature schemes exist, assuming one-way functions~\cite{Goldreich2004}.

\paragraph{Probabilistic Checkable Proofs.} Another ingredient required in the construction is the Probabilistic Checkable Proofs (PCPs) under Levin reductions. We state this only for Max-3SAT. 

PCPs with $\cFthreesat$ predicates consist of three (deterministic) algorithms:
\begin{itemize}
\item $\arr$ takes in the circuit $C$ and outputs a set $\tcF \subseteq \cFthreesat_d$ of predicates
\item $\enc$ takes in any circuit $C$ and an input $y$ for the circuit $C$, and produces $\pi \in \{0, 1\}^d$,
\item $\dec$ takes in $\pi \in \{0, 1\}^d$ and $C$ and produces an input $y'$ for the circuit $C$.
\end{itemize}

The guarantees are as follows:
\begin{itemize}
\item (Completeness) If $C(y) = 1$, then $\psi(\enc(y, C)) = 1$ for all $\psi \in \tcF$,
\item (Soundness) If $\E_{\psi \sim \tcF}[\psi(\pi)] \geq \gamma$, then $C(\dec(\pi, C)) = 1$.
\end{itemize}

In the papers that originally proved the PCP theorem~\cite{AroraLMSS98,AroraS98}, it was also shown that PCPs of the above form exist for some constant $\gamma \in (0, 1)$. 

\subsection{The Construction and Hardness Argument}

Formally, the construction of \cite{UllmanV20} (which builds on \cite{dwork2009complexity}) works as follows:
\begin{itemize}
\item Let $\sk, \vk$ be the key pair generated by $\Gen$.
\item Let $\sigma_1, \dots, \sigma_n$ be outputs from $n$ runs of $\Sign(\sk)$.
\item Let $C_{\vk}$ denote the circuit for $\Ver(\vk, \cdot)$.
\item Let $\tcF = \arr(C_{\vk})$.
\item For $i \in [n]$, let $\pi_i = \enc(\sigma_i, C_{\vk})$.
\item Output the dataset $D = (\pi_1, \dots, \pi_n)$ together with $\tcF$.
\end{itemize}

Now, suppose for the sake of contradiction that there is an $(\eps, 1/n^{\omega(1)})$-DP $(\cFthreesat, \gamma)$-promise sanitizer that runs in polynomial time. Then, we can run the sanitizer on $D$ to produce a synthetic dataset $\tD = (\tx_1, \dots, \tx_{\tn})$. Now, from the completeness of the PCP, we have
$\E_{\psi \sim \tcF}\left[\psi(D)\right] = 1$. Thus, by the guarantee of the promise sanitizer, with probability $2/3$, we have $\E_{\psi \sim \tcF}\left[\psi(\tD)\right] \geq \gamma$. When this happens, there must exist $\tx_i \in \tD$ such that $\E_{\psi \sim \tcF}\left[\psi(\tx_i)\right] \geq \gamma$. Finally, we can run $\dec(\tx_i, C_\vk)$ to get a new signature $\sigma^*$. 

By the guarantee of the PCP, we have $C_{\vk}(\sigma^*) = 1$, i.e., $\Ver(\vk, \sigma^*) = 1$. Now, one of the following two cases can happen: (i) if $\sigma^* \notin \{\sigma_1, \dots, \sigma_n\}$, then we violate the security of the signature scheme, or (ii) if $\sigma^* \in \{\sigma_1, \dots, \sigma_n\}$, then $\pi^* = \enc(\sigma^*, C_{\vk})$ belongs to the input dataset and we violate $(\eps, 1/n^{\omega(1)})$-DP.